\documentclass[letter,romanappendices,onecolumn,draft]{ieeeconf}
\let\proof\relax   

\usepackage{amsthm,xpatch}
\usepackage{amsmath,amsfonts}
\usepackage{cite}
\usepackage{amssymb}
\usepackage{dsfont}
\usepackage{graphicx, subfigure}
\usepackage{color}
\usepackage{breqn}
\usepackage{mathtools}
\usepackage{bbm}
\usepackage{latexsym}
\usepackage[ruled, linesnumbered]{algorithm2e}
\usepackage{accents}
\usepackage{tikz}
\usepackage[mathscr]{euscript}
\usepackage{bm}
\usepackage{comment}

\newtheorem{lemma}{Lemma}
\newtheorem{theorem}{Theorem}
\newtheorem{remark}{Remark}

\newtheorem{example}{Example}

\newcommand*{\transpose}{%
  {\mathpalette\@transpose{}}%
}

\IEEEoverridecommandlockouts

\begin{document}

\newcommand{\SB}[3]{
\sum_{#2 \in #1}\biggl|\overline{X}_{#2}\biggr| #3
\biggl|\bigcap_{#2 \notin #1}\overline{X}_{#2}\biggr|
}

\newcommand{\Mod}[1]{\ (\textup{mod}\ #1)}

\newcommand{\overbar}[1]{\mkern 0mu\overline{\mkern-0mu#1\mkern-8.5mu}\mkern 6mu}

\makeatletter
\newcommand*\nss[3]{%
  \begingroup
  \setbox0\hbox{$\m@th\scriptstyle\cramped{#2}$}%
  \setbox2\hbox{$\m@th\scriptstyle#3$}%
  \dimen@=\fontdimen8\textfont3
  \multiply\dimen@ by 4             
  \advance \dimen@ by \ht0
  \advance \dimen@ by -\fontdimen17\textfont2
  \@tempdima=\fontdimen5\textfont2  
  \multiply\@tempdima by 4
  \divide  \@tempdima by 5          
  \ifdim\dimen@<\@tempdima
    \ht0=0pt                        
    \@tempdima=\fontdimen5\textfont2
    \divide\@tempdima by 4          
    \advance \dimen@ by -\@tempdima 
    \ifdim\dimen@>0pt
      \@tempdima=\dp2
      \advance\@tempdima by \dimen@
      \dp2=\@tempdima
    \fi
  \fi
  #1_{\box0}^{\box2}%
  \endgroup
  }
\makeatother

\makeatletter
\renewenvironment{proof}[1][\proofname]{\par
  \pushQED{\qed}%
  \normalfont \topsep6\p@\@plus6\p@\relax
  \trivlist
  \item[\hskip\labelsep
        \itshape
    #1\@addpunct{:}]\ignorespaces
}{%
  \popQED\endtrivlist\@endpefalse
}
\makeatother

\makeatletter
\newsavebox\myboxA
\newsavebox\myboxB
\newlength\mylenA

\newcommand*\xoverline[2][0.75]{%
    \sbox{\myboxA}{$\m@th#2$}%
    \setbox\myboxB\null
    \ht\myboxB=\ht\myboxA%
    \dp\myboxB=\dp\myboxA%
    \wd\myboxB=#1\wd\myboxA
    \sbox\myboxB{$\m@th\overline{\copy\myboxB}$}
    \setlength\mylenA{\the\wd\myboxA}
    \addtolength\mylenA{-\the\wd\myboxB}%
    \ifdim\wd\myboxB<\wd\myboxA%
       \rlap{\hskip 0.5\mylenA\usebox\myboxB}{\usebox\myboxA}%
    \else
        \hskip -0.5\mylenA\rlap{\usebox\myboxA}{\hskip 0.5\mylenA\usebox\myboxB}%
    \fi}
\makeatother

\xpatchcmd{\proof}{\hskip\labelsep}{\hskip3.75\labelsep}{}{}

\pagestyle{plain}

\title{\fontsize{19}{27}\selectfont Private Linear Transformation: The Individual Privacy Case}

\author{Nahid Esmati, Anoosheh Heidarzadeh, and Alex Sprintson\thanks{The authors are with the Department of Electrical and Computer Engineering, Texas A\&M University, College Station, TX 77843 USA (E-mail: \{nahid,anoosheh, spalex\}@tamu.edu).}}


%


\maketitle 

\thispagestyle{plain}

\begin{abstract}


This paper considers the single-server Private Linear Transformation (PLT) problem when individual privacy is required. In this problem, there is a user that wishes to obtain $L$ linear combinations of a $D$-subset of messages belonging to a dataset of $K$ messages stored on a single server. The goal is to minimize the download cost while keeping the identity of every message required for the computation individually private. The individual privacy requirement implies that, from the perspective of the server, every message is equally likely to belong to the $D$-subset of messages that constitute the support set of the required linear combinations. We focus on the setting in which the matrix of coefficients pertaining to the required linear combinations is the generator matrix of a Maximum Distance Separable code. We establish lower and upper bounds on the capacity of PLT with individual privacy, where the capacity is defined as the supremum of all achievable download rates. We show that our bounds are tight under certain divisibility conditions. In addition, we present lower bounds on the capacity of the settings in which the user has a prior side information about a subset of messages. 
\end{abstract}

\section{introduction}

In this work, 
we study the problem of single-server 
\emph{Private Linear Transformation (PLT) with Individual Privacy}, referred to as \emph{IPLT} for short. In this problem, 
there is a single server that stores a set of $K$ messages, and a user that wants to compute $L$ linear combinations of a subset of $D$ messages. The objective of the user is to recover the required linear combinations by downloading minimum possible amount of information from the server, while 
protecting the privacy of the identity of every message required for the computation individually.
The individual privacy requirement implies that, from the server's perspective, every message is equally likely to belong to the $D$-subset of messages that constitute the support set of the required linear combinations.

The PLT problem is motivated by the application of \emph{linear transformation} for dimensionality reduction in Machine Learning (ML), see, e.g.,~\cite{CG2015}. 
Consider a dataset with $N$ data samples, each with $K$ attributes. 
Consider a user that wishes to implement an ML  algorithm on a subset of $D$ selected attributes, while protecting the privacy of the selected attributes.  
When $D$ is large, the $D$-dimensional feature space is typically mapped onto a new subspace of lower dimension, say, $L$, and the ML algorithm operates on the new $L$-dimensional subspace instead. 
A commonly-used technique for dimensionality reduction is \emph{linear transformation}, where an $L\times D$ matrix is multiplied by the $D\times N$ data submatrix (the submatrix of the original $K\times N$ data matrix restricted to the $D$ selected attributes). 
Thinking of the rows of the $K\times N$ data matrix as the $K$ messages, the labels of the $D$ selected attributes as the identities of the $D$ messages in the support set of the required linear combinations, and the $L\times D$ matrix used for transformation as the coefficient matrix of the required linear combinations, this scenario matches the setup of the PLT problem. In addition, when the privacy requirement is to hide the identity of each of the $D$ selected attributes individually, the problem reduces to PLT with individual privacy. 


Individual privacy was originally introduced in~\cite{HKRS2019} for single-server Private Information Retrieval (PIR) with Individual Privacy (IPIR), and was recently considered for single-server Private Linear Computation (PLC) with Individual Privacy (IPLC) in~\cite{HS2020}. 
The IPLT problem generalizes the IPIR and IPLC problems. In particular, IPLT reduces to IPIR or IPLC for $L=D$ or $L=1$, respectively. 
The IPLT problem is also related to the single-server PLT with Joint Privacy (JPLT) problem~
\cite{EHS2021Joint}. In JPLT, 
the identities of all $D$ messages required for the computation must be kept jointly private. 
The joint privacy requirement implies that, from the perspective of the server, any $D$-subset of messages is equally likely to be the support set of the required linear combinations. 
Joint privacy was previously considered in~\cite{BU2018,HKGRS2018,LG2018} for PIR, and  in~\cite{HS2019PC,HS2020} for PLC. 
Despite the fact that joint privacy is a stronger notion of privacy, 
individual privacy may still provide a satisfactory degree of privacy 
in several practical scenarios (see, e.g.,~\cite{KKHS32019}). 
It should be noted that both joint and individual privacy conditions are relaxed versions of the privacy condition considered in~\cite{SJ2018,MM2018,OK2018,OLRK2018} for multi-server PLC, where 
the values of the combination coefficients in the required linear combination must be kept private. 

In~\cite{HKRS2019} and~\cite{HS2020}, IPIR and IPLC were studied in the settings in which the user has a prior side information about a subset of messages. It was shown that, when compared to PIR and PLC with joint privacy, both IPIR and IPLC can be performed with a much lower download cost. 
Motivated by these results, this work seeks to answer the following questions: 
(i) is it possible to perform IPLT with a lower download cost than JPLT? 
(ii) is it possible to leverage a prior side information in order to further decrease the download cost of IPLT?
(iii) what are the fundamental limits on the download cost for IPLT? 
We make a significant progress towards answering these questions in this work.

\subsection{Main Contributions}


In this work, we focus on the setting in which the coefficient matrix corresponding to the required linear combinations is the generator matrix of a Maximum Distance Separable (MDS) code. 
The MDS coefficient matrices are motivated by the application of \emph{random linear transformation} for dimensionality reduction (see, e.g.,~\cite{BM2001}), where a random $L\times D$ matrix is used for transformation. 
Note that 
an $L\times D$ matrix whose entries are randomly chosen from a sufficiently large field is MDS with high probability.

For this setting, we establish lower and upper bounds on the capacity of IPLT, where the capacity is defined as the supremum of all achievable download rates. In addition, we show that our bounds are tight under certain divisibility conditions, settling the capacity of IPLT for such cases. 
To prove the upper bound on the capacity, we use  information-theoretic arguments based on a necessary condition for IPLT schemes, and formulate the problem as an integer linear programming (ILP) problem. Solving this ILP, we obtain the capacity upper bound. The lower bound on the capacity is proven by a novel achievability scheme, termed \emph{Generalized Partition-and-Code with Partial Interference Alignment (GPC-PIA) protocol}. This protocol generalizes the protocols we previously proposed in~\cite{HKRS2019} and~\cite{HS2020} for IPIR and IPLC, respectively. 
In addition, we present lower bounds on the capacity of the settings in which the user has a prior side information about a subset of messages. 
Our results indicate that, when there is no side information, IPLT can be performed more efficiently than JPLT, in terms of the download cost. The advantage of IPLT over JPLT is even more pronounced 
when the user knows a subset of messages or a subspace spanned by them as side information. 







\section{Problem Setup}\label{sec:SN}




Throughout this paper, 
random variables and their realizations are denoted by bold-face symbols (e.g., $\mathbf{X},\mathbf{W}$) and non-bold-face symbols (e.g., $\mathrm{X},\mathrm{W}$), respectively. 

Let $\mathbb{F}_p$ be a finite field of order $p$, and let $\mathbb{F}_{q}$ be an extension field of $\mathbb{F}_p$. 
Let $K,D,L$ be positive integers such that ${L\leq D\leq K}$, and let $\mathcal{K}\triangleq \{1,...,K\}$. We denote by $\mathscr{W}$ the set of all $D$-subsets of $\mathcal{K}$, and denote by $\mathscr{V}$ the set of all $L\times D$ matrices (with entries from $\mathbb{F}_p$) that generate a Maximum Distance Separable (MDS) code. 


Suppose that there is a server that stores $K$ messages ${X_1,\dots,X_K}$, where $X_i\in \mathbb{F}_q$ for $i\in \mathcal{K}$. 
Let $\mathrm{X}\triangleq [X_1,\dots,X_K]^{\mathsf{T}}$.
For every ${\mathrm{S}\subset \mathcal{K}}$, we denote by $\mathrm{X}_{\mathrm{S}}$ the vector $\mathrm{X}$ restricted to its components indexed by $\mathrm{S}$. 
We assume that $\mathbf{X}_1,\dots,\mathbf{X}_K$ are independently and uniformly distributed over $\mathbb{F}_{q}$. That is, ${H(\mathbf{X}_{i})=\theta\triangleq \log_2 q}$ for $i\in \mathcal{K}$, and ${H(\mathbf{X}_{\mathrm{S}})= |\mathrm{S}| \theta}$ for ${\mathrm{S}\subset \mathcal{K}}$, where $|\mathrm{S}|$ denotes the size of $\mathrm{S}$. 
Note that $H(\mathbf{X})=K\theta$. 
Suppose that there is a user who wants to compute 
the vector $\mathrm{Z}^{[\mathrm{W},\mathrm{V}]}\triangleq \mathrm{V}\mathrm{X}_{\mathrm{W}}$, where ${\mathrm{W}\in \mathscr{W}}$ and ${\mathrm{V}\in \mathscr{V}}$. 
That is, $\mathrm{Z}^{[\mathrm{W},\mathrm{V}]}$ contains $L$ components $\mathrm{v}^{\mathsf{T}}_1 \mathrm{X}_{\mathrm{W}},\dots,\mathrm{v}^{\mathsf{T}}_L \mathrm{X}_{\mathrm{W}}$, where $\mathrm{v}_l^{\mathsf{T}}$ is the $l$th row of $\mathrm{V}$. 
Note that $H(\mathbf{Z}^{[\mathrm{W},\mathrm{V}]})=L\theta$. We refer to $\mathrm{Z}^{[\mathrm{W},\mathrm{V}]}$ as the \emph{demand}, $\mathrm{W}$ as the \emph{support index set of the demand}, $\mathrm{V}$ as the \emph{coefficient matrix of the demand}, $D$ as the \emph{support size of the demand}, and $L$ as the \emph{dimension of the demand}.

We assume that (i) $\mathbf{W}$, $\mathbf{V}$, and $\mathbf{X}$ are independent; (ii) $\mathbf{W}$ and $\mathbf{V}$ are uniformly distributed over ${\mathscr{W}}$ and $\mathscr{V}$, respectively; and (iii) the parameters $D$ and $L$, and the distribution of $(\mathbf{W},\mathbf{V})$ are initially known by the server, whereas the realization $(\mathrm{W},\mathrm{V})$ is not initially known by the server.



Given the realization $(\mathrm{W},\mathrm{V})$, the user generates a query $\mathrm{Q} = \mathrm{Q}^{[\mathrm{W},\mathrm{V}]}$, which is a (potentially stochastic) function of $(\mathrm{W},\mathrm{V})$, and sends it to the server. 
The query $\mathrm{Q}$ 
must satisfy the following privacy condition: given the query $\mathrm{Q}$, 
every message index must be equally likely to be in the demand's support index set. 
That is, for every $i\in \mathcal{K}$, it must hold that
\begin{equation*}
\Pr (i\in \mathbf{W}|\mathbf{Q}=\mathrm{Q})=\Pr(i \in \mathbf{W}), 
\end{equation*} where $\mathbf{Q}$ denotes $\mathbf{Q}^{[\mathbf{W},\mathbf{V}]}$. 
This condition---which was recently introduced in~\cite{HKRS2019} and~\cite{HS2020} for single-server PIR and PLC, is referred to as the \emph{individual privacy condition}.

Upon receiving the query $\mathrm{Q}$, the server generates an answer $\mathrm{A}=\mathrm{A}^{[\mathrm{W},\mathrm{V}]}$, 
and sends it back to the user. 
The answer $\mathrm{A}$ is a deterministic function of $\mathrm{Q}$ and $\mathrm{X}$.
The collection of the answer $\mathrm{A}$, the query $\mathrm{Q}$, and the realization $(\mathrm{W}, \mathrm{V})$ must enable the user to recover the demand $\mathrm{Z}^{[\mathrm{W},\mathrm{V}]}$. That is, 
${H(\mathbf{Z}| \mathbf{A},\mathbf{Q}, \mathbf{W},\mathbf{V})=0}$, where $\mathbf{Z}$ and $\mathbf{A}$ denote $\mathbf{Z}^{[\mathbf{W},\mathbf{V}]}$ and $\mathbf{A}^{[\mathbf{W},\mathbf{V}]}$, respectively. This condition is referred to as the \emph{recoverability condition}. 

 
We would like to design a protocol for generating a query $\mathrm{Q}^{[\mathrm{W},\mathrm{V}]}$ and the corresponding answer $\mathrm{A}^{[\mathrm{W},\mathrm{V}]}$ 
such that both the individual privacy and recoverability conditions are satisfied.
We refer to this problem as single-server \emph{Private Linear Transformation (PLT) with Individual Privacy}, or \emph{IPLT} for short. 
We define the \emph{rate} of an IPLT protocol as the ratio of the entropy of the demand (i.e., $H(\mathbf{Z})=L\theta$) to the entropy of the answer (i.e., $H(\mathbf{A})$). We also define the \emph{capacity} of the IPLT setting as the supremum of rates over all IPLT protocols.

In this work, our goal is to establish (preferably matching) non-trivial lower and upper bounds (in terms of the parameters $K$, $D$, and $L$) on the capacity of the IPLT setting. 

\section{A Necessary Condition for IPLT Protocols}
The individual privacy and recoverability conditions yield a necessary (but not sufficient) condition for any IPLT protocol, stated in Lemma~\ref{lem:NCIPLT}. The proof is straightforward by the way of contradiction, and hence, omitted. 

\begin{lemma}\label{lem:NCIPLT}
Given any IPLT protocol, for any $i\in\mathcal{K}$, there must exist $\mathrm{W}^{*}\in\mathscr{W}$ with $i\in \mathrm{W}^{*}$, and $\mathrm{V}^{*}\in\mathscr{V}$, such that \[H(\mathbf{Z}^{[\mathrm{W}^{*},\mathrm{V}^{*}]}| \mathbf{A}, \mathbf{Q}) = 0.\] 	
\end{lemma}



The result of Lemma~\ref{lem:NCIPLT} establishes a connection between \emph{linear codes} with a certain constraint and \emph{linear schemes} for IPLT, i.e., any scheme in which the server's answer to the user's query consists of only \emph{linear} combinations of the messages. 
In particular, the matrix of combination coefficients---pertaining to the linear combinations in the answer, must be the generator matrix of a linear code of length $K$ that satisfies the following condition: for any coordinate $i$, there must exist $K-D$ coordinates different from $i$ such that the code resulting from puncturing\footnote{To puncture a linear code at a coordinate, the column corresponding to that coordinate is deleted from the generator matrix of the code.} at these $K-D$ coordinates contains $L$ codewords that are MDS. Note, however, that this condition is only necessary and not sufficient. In particular, a sufficient (yet not necessary) condition is that, for \emph{every coordinate $i$}, the punctured codes resulting from puncturing at any $K-D$ other coordinates (different from $i$) collectively contain the \emph{same number of groups} of $L$ codewords that are MDS. Maximizing the rate of a linear IPLT scheme is then equivalent to minimizing the dimension of a linear code that satisfies this sufficient condition. 
However, despite the fact that this sufficient condition is stronger than the necessary condition provided by Lemma~\ref{lem:NCIPLT}, the former is more combinatorial, while the latter is more information-theoretic and hence more useful in the converse proof. 

\section{Main Results}
This section summarizes our main results for IPLT. 

\begin{theorem}\label{thm:IPLT}
For the IPLT setting with $K$ messages, demand's support size $D$, and demand's dimension $L$, the capacity is lower and upper bounded by ${(\lfloor \frac{K}{D}\rfloor+\min\{\frac{R}{S},\frac{R}{L}\})^{-1}}$ and $(\lfloor \frac{K}{D}\rfloor+\min\{1,\frac{R}{L}\})^{-1}$, respectively, where $R\triangleq K \pmod D$ and $S\triangleq \gcd(D+R,R)$. The lower and upper bounds match when $R\leq L$ or $R\mid D$. 	
\end{theorem}

To prove the converse bound, we use the necessary condition for IPLT protocols provided by Lemma~\ref{lem:NCIPLT} along with information-theoretic arguments, and formulate the problem as an integer linear programming (ILP) problem. 
Solving this ILP, we obtain the upper bound on the capacity (see Section~\ref{sec:IPLT-Conv}). 
The lower bound on the capacity is proven by constructing an IPLT protocol, called \emph{Generalized Partition-and-Code with Partial Interference Alignment (GPC-PIA)} (see Section~\ref{sec:IPLT-Ach}). 
This protocol is a generalization of the protocols we previously proposed in~\cite{HKRS2019} and~\cite{HS2020} for single-server PIR and PLC (without SI) with individual privacy. 
The main ingredients of the GPC-PIA protocol are as follows: 
(i) constructing a properly designed family of subsets of messages probabilistically, where some subsets are possibly overlapping, and 
(ii) designing a number of linear combinations for each subset judiciously, where the linear combinations pertaining to the overlapping subsets are partially aligned.

\begin{remark}\label{rem:IPLT1}
\emph{As shown in~\cite{HS2020}, the capacity of PLC with individual privacy, which is a special case of IPLT for $L=1$, is given by $\lceil \frac{K}{D+M}\rceil^{-1}$, where the user initially knows $M$ uncoded messages or one linear combination of $M$ messages as side information. 
The capacity of this setting was left open for $M=0$. 
Theorem~\ref{thm:IPLT} provides a lower bound ${(\lfloor \frac{K}{D}\rfloor+\min\{\frac{R}{S},R\})^{-1}}$ and an upper bound ${(\lfloor \frac{K}{D}\rfloor+\min\{1,R\})^{-1}}$ on the capacity of this setting. 
Interestingly, these bounds are matching when $R=0$ or $R\mid D$, settling the capacity of PLC (without SI) with individual privacy for these cases. 
For $L=D$, IPLT reduces to PIR (without SI) with individual privacy, and an optimal scheme in this case 
is to download the entire dataset~\cite{HKRS2019}.}
\end{remark}

\begin{remark}\label{rem:JPLT2}
\emph{The result of Theorem~\ref{thm:IPLT} can be extended to IPLT with Side Information (SI). 
We consider two types of SI---previously studied in the PIR and PLC literature: \emph{Uncoded SI (USI)} (see~\cite{KGHERS2020}), and \emph{Coded SI (CSI)} (see~\cite{HKS2019Journal}).
In the case of USI (or CSI), the user initially knows a subset of $M$ messages (or $L$ MDS coded combinations of $M$ messages). For both USI and CSI, the identities of these $M$ messages 
are initially unknown by the server. 
When the identity of every message in the support sets of demand and side information must be protected individually, a slightly modified version of the GPC-PIA scheme (for IPLT without SI) achieves the rate ${(\lfloor \frac{K}{D+M}\rfloor+\min\{\frac{R}{S},\frac{R}{L}\})^{-1}}$ for both IPLT with USI and CSI, where ${R= K \pmod {D+M}}$ and ${S= \gcd(D+M+R,R)}$. This result generalizes the results of~\cite{HKRS2019} and~\cite{HS2020} for PIR and PLC with individual privacy. The optimality of this rate, however, remains open in general.}
\end{remark}

\section{Proof of Converse}\label{sec:IPLT-Conv}


\begin{lemma}\label{lem:IPLT-Conv}
The rate of any IPLT protocol for $K$ messages, demand's support size $D$ and dimension $L$, is upper bounded by $(\lfloor \frac{K}{D}\rfloor+\min\{1,\frac{R}{L}\})^{-1}$, where ${R\triangleq K \pmod D}$.
\end{lemma}

\begin{proof}
Consider an arbitrary IPLT protocol that generates the query-answer pair $(\mathrm{Q}^{[\mathrm{W},\mathrm{V}]},\mathrm{A}^{[\mathrm{W},\mathrm{V}]})$ for any given $(\mathrm{W},\mathrm{V})$. 
To prove the rate upper bound in the lemma, we need to show that ${H(\mathbf{A})\geq (L\lfloor \frac{K}{D}\rfloor+\min\{L,R\})\theta}$. 
Recall that  $\mathbf{A}$ denotes $\mathbf{A}^{[\mathbf{W},\mathbf{V}]}$, and $\theta$ is the entropy of a message. 
Consider an arbitrary message index $k_1\in \mathcal{K}$. 
By the result of Lemma~\ref{lem:NCIPLT}, there exist $\mathrm{W}_1\in \mathscr{W}$ with $k_1\in \mathrm{W}_1$, and $\mathrm{V}_1\in\mathscr{V}$ such that $H(\mathbf{Z}_1|\mathbf{A},\mathbf{Q})=0$, where $\mathbf{Z}_1\triangleq \mathbf{Z}^{[\mathrm{W}_1,\mathrm{V}_1]}$. 
By the same arguments as in the proof of \cite[Lemma~2]{EHS2021Joint}, we have
\begin{align}
H(\mathbf{A})&\geq H(\mathbf{A}|\mathbf{Q})+H(\mathbf{Z}_1|\mathbf{A},\mathbf{Q}) \nonumber \\
& = {H(\mathbf{Z}_1|\mathbf{Q})+H(\mathbf{A}|\mathbf{Q},\mathbf{Z}_1}) \nonumber \\ 
& = {H(\mathbf{Z}_1)+H(\mathbf{A}|\mathbf{Q},\mathbf{Z}_1}) \label{eq:5}   
\end{align}
To further lower bound $H(\mathbf{A}|\mathbf{Q},\mathbf{Z}_1)$, we proceed as follows. Take an arbitrary message index $k_2\not\in \mathrm{W}_1$.
Again, by Lemma~\ref{lem:NCIPLT}, there exist $\mathrm{W}_2\in \mathscr{W}$ with $k_2\in \mathrm{W}_2$, and $\mathrm{V}_2\in\mathscr{V}$ such that $H(\mathbf{Z}_2|\mathbf{A},\mathbf{Q})=0$, where $\mathbf{Z}_2\triangleq \mathbf{Z}^{[\mathrm{W}_2,\mathrm{V}_2]}$. 
Using a similar technique as in~\eqref{eq:5}, it follows that $H(\mathbf{A}|\mathbf{Q},\mathbf{Z}_1)\geq  H(\mathbf{Z}_2|\mathbf{Q},\mathbf{Z}_1)+H(\mathbf{A}|\mathbf{Q},\mathbf{Z}_1,\mathbf{Z}_2)$, and consequently,  
\begin{equation}\label{eq:6}
H(\mathbf{A}|\mathbf{Q},\mathbf{Z}_1)\geq H(\mathbf{Z}_2|\mathbf{Z}_1)+H(\mathbf{A}|\mathbf{Q},\mathbf{Z}_2,\mathbf{Z}_1).    
\end{equation}
Combining~\eqref{eq:5} and~\eqref{eq:6}, we get 
\begin{equation}\label{eq:7}
H(\mathbf{A})\geq H(\mathbf{Z}_1)+H(\mathbf{Z}_2|\mathbf{Z}_1)+H(\mathbf{A}|\mathbf{Q},\mathbf{Z}_2,\mathbf{Z}_1).    
\end{equation} 



We repeat this process multiple rounds until there is no message index left to take. 
Let $n$ be the total number of rounds, and let $k_1,\dots,k_n$ be the message indices chosen over the rounds. 
For every $i\in \{1,\dots,n\}$, let $\mathrm{W}_i\in \mathscr{W}$ with $k_i\in \mathrm{W}_i$ and $k_i\not\in \cup_{1\leq j<i}\mathrm{W}_j$, and $\mathrm{V}_i\in \mathscr{V}$, be such that ${H(\mathbf{Z}_{i}|\mathbf{A},\mathbf{Q})=0}$, where $\mathbf{Z}_{i}\triangleq \mathbf{Z}^{[\mathrm{W}_i,\mathrm{V}_i]}$. 
Note that $\cup_{1\leq i\leq n}\mathrm{W}_i = \mathcal{K}$. 
Similarly as before, we can show that
\begin{align}\label{eq:8}
H(\mathbf{A}) & \geq \sum_{1\leq i\leq n} H(\mathbf{Z}_i|\mathbf{Z}_{i-1},\dots,\mathbf{Z}_{1}) \nonumber \\ 
& \quad +H(\mathbf{A}|\mathbf{Q},\mathbf{Z}_n,\dots,\mathbf{Z}_1)\nonumber\\
&\geq \sum_{1\leq i\leq n} H(\mathbf{Z}_i|\mathbf{Z}_{i-1},\dots,\mathbf{Z}_{1}).
\end{align} 
Let $\mathbf{Z}_{i,1},\dots,\mathbf{Z}_{i,L}$ be the components of $\mathbf{Z}_i$, where $\mathbf{Z}_{i,l}\triangleq \mathrm{v}^{\mathsf{T}}_{i,l} \mathbf{X}_{\mathrm{W}_i}$, and $\mathrm{v}^{\mathsf{T}}_{i,l}$ is the $l$th row of $\mathrm{V}_i$.
Next, we show that 
\begin{equation}\label{eq:9}
H(\mathbf{Z}_i|\mathbf{Z}_{i-1},\dots,\mathbf{Z}_{1})\geq \min\{N_i,L\}\theta,    
\end{equation}
where ${N_i\triangleq |\mathrm{W}_i\setminus \cup_{1\leq j<i}\mathrm{W}_j|}$ is the number of message indices that belong to $\mathrm{W}_i$, but not $\cup_{1\leq j<i}\mathrm{W}_j$. (Note that ${N_1= |\mathrm{W}_1|=D}$.) 
This is equivalent to showing that $\mathbf{Z}_i$ contains $M_i\triangleq \min\{N_i,L\}$ components that are independent of the components of $\mathbf{Z}_1,\dots,\mathbf{Z}_{i-1}$.
Note that the components of $\mathbf{Z}_1,\dots,\mathbf{Z}_i$ are linear combinations of the messages $\mathbf{X}_1,\dots,\mathbf{X}_K$. 
Let $\mathrm{u}_{i,l}$ be a column-vector of length $K$ such that the vector $\mathrm{u}_{i,l}$ restricted to its components indexed by $\mathrm{W}_i$ is the vector $\mathrm{v}_{i,l}$, and the rest of the components of the vector $\mathrm{u}_{i,l}$ are all zero, and let  
${\mathrm{U}_i \triangleq [\mathrm{u}_{i,1},\dots,\mathrm{u}_{i,L}]^{\mathsf{T}}}$. Thus, we need to show that the matrix $\mathrm{U}_i$ contains $M_i$ rows that are linearly independent of the rows of the matrices $\mathrm{U}_1,\dots,\mathrm{U}_{i-1}$. 
Note that the rows of the matrix $\mathrm{U}_i$ are linearly independent, because $\mathrm{U}_i$ contains $\mathrm{V}_i$ as a submatrix, and $\mathrm{V}_i$ is a full-rank matrix. 

Let $\mathrm{S}_i$ be an $L\times N_i$ submatrix of $\mathrm{U}_i$ formed by columns indexed by $\mathrm{W}_i\setminus \cup_{1\leq j<i}\mathrm{W}_j$. 
Note that $\mathrm{S}_i$ is a submatrix of $\mathrm{V}_i$, and every $L\times L$ submatrix of $\mathrm{V}_i$ is full rank (by assumption). 
We consider two cases: (i) $N_i\leq L$, and (ii) $N_i>L$. 
In the case (i), the $N_i$ columns of $\mathrm{S}_i$ are linearly independent. Otherwise, any $L\times L$ submatrix of $\mathrm{V}_i$ that contains $\mathrm{S}_i$ cannot be full rank, and hence a contradiction. 
In the case (ii), any $L$ columns of $\mathrm{S}_i$ are linearly independent. 
Otherwise, $\mathrm{S}_i$ (and consequently, $\mathrm{V}_i$) contains an $L\times L$ submatrix that is not full rank, which is again a contradiction. 
By these arguments, the rank of $\mathrm{S}_i$ is $M_i=\min\{L,N_i\}$, and hence, $\mathrm{S}_i$ contains $M_i$ linearly independent rows. 

Without loss of generality, assume that the first $M_i$ rows of $\mathrm{S}_i$ are linearly independent. 
Moreover, the submatrix of $[\mathrm{U}_{1};\dots;\mathrm{U}_{i-1}]$ restricted to its columns indexed by ${\mathrm{W}_i\setminus \cup_{1\leq j<i}\mathrm{W}_j}$ (and all its rows) 
is an all-zero matrix, where $[\mathrm{U}_{1};\dots;\mathrm{U}_{i-1}]$ is a matrix formed by stacking $\mathrm{U}_1,\dots,\mathrm{U}_{i-1}$ vertically. 
By these arguments, it follows that the first $M_i$ rows of $\mathrm{U}_i$ are linearly independent of the rows of $[\mathrm{U}_{1};\dots;\mathrm{U}_{i-1}]$. This completes the proof of~\eqref{eq:9}. 

Combining~\eqref{eq:8} and~\eqref{eq:9}, we have 
\begin{equation}\label{eq:10}
H(\mathbf{A})\geq \sum_{1\leq i\leq n} \min\{L,N_i\}\theta    
\end{equation} 
Recall that $N_i = |\mathrm{W}_i\setminus \cup_{1\leq j<i} \mathrm{W}_j|$. 
Note that ${1\leq N_i\leq D}$ since $\mathrm{W}_i\setminus \cup_{1\leq j<i} \mathrm{W}_j$ is a subset of $\mathrm{W}_i$, and the message index $k_i$ belongs to ${\mathrm{W}_i\setminus \cup_{1\leq j<i} \mathrm{W}_j}$. 
Moreover, ${\sum_{i=1}^{n} N_i = K}$ since $\mathrm{W}_1$, ${\mathrm{W}_2\setminus \mathrm{W}_1}$, $\dots$, ${\mathrm{W}_n\setminus \cup_{1\leq j<n} \mathrm{W}_j}$ form a partition of $\mathcal{K}$, and ${|\mathrm{W}_1|=N_1=D}$, ${|\mathrm{W}_2\setminus \mathrm{W}_1|=N_2}$, $\dots$, ${|\mathrm{W}_n\setminus \cup_{1\leq j<n} \mathrm{W}_j|=N_n}$. 
To obtain a converse bound, we need to minimize $\sum_{1\leq i\leq n} \min\{L,N_i\}$, subject to the constraints (i) ${N_1=D}$, and $1\leq N_i\leq D$ for any ${1<i\leq n}$, and (ii) ${\sum_{1\leq i\leq n} N_i = K}$. 
To this end, we reformulate this optimization problem as follows. 

For every ${j\in \{1,\dots,D\}}$, let $T_j\triangleq \sum_{1\leq i\leq n} \mathbbm{1}_{\{N_i=j\}}$ be the number of rounds $i$ such that ${N_i=j}$. 
Using this notation, the objective function $\sum_{1\leq i\leq n} \min\{L,N_i\}$ can be rewritten as $\sum_{1\leq j\leq D}  T_j\min\{L,j\}$, or equivalently, ${\sum_{1\leq j\leq L} T_j j+ \sum_{L<j\leq D} T_j L}$; the constraint (i) reduces to $T_j\in \mathbb{N}_0\triangleq\{0,1,\dots\}$ for every $1\leq j<D$, and ${T_D\in \mathbb{N}\triangleq \{1,2,\dots\}}$; and the constraint (ii) reduces to $\sum_{1\leq j\leq D} T_j j = K$. Thus, we need to solve the following integer linear programming (ILP) problem: 
\begin{eqnarray}\label{eq:11}
& \hspace{-1cm} \mathrm{minimize} & \sum_{1\leq j\leq L} T_j j+ \sum_{L<j\leq D} T_j L\\ \nonumber
& \hspace{-1cm} \mathrm{subject~to} & \sum_{1\leq j\leq D} T_j j= K\\ \nonumber
&& T_j\in \mathbb{N}_0, \quad \forall 1\leq j<D\\ \nonumber
&& T_D\in \mathbb{N}.
\end{eqnarray}
Solving this ILP using the Gomory's cutting-plane algorithm~\cite{MMWW2002}, an optimal solution is $T_D = \lfloor\frac{K}{D}\rfloor$, $T_R = 1$, and $T_j=0$ for all ${j\not\in\{R,D\}}$, where ${R\triangleq K \pmod D}$, and the optimal value is ${L\lfloor\frac{K}{D} \rfloor+\min\{L,R\}}$. Equivalently, $\sum_{1\leq i\leq n} \min\{L,N_i\}\geq L\lfloor\frac{K}{D} \rfloor+\min\{L,R\}$. Combining this inequality and the inequality~\eqref{eq:10}, we have ${H(\mathbf{A})\geq (L\lfloor\frac{K}{D} \rfloor+\min\{L,R\})\theta}$, as was to be shown.
\end{proof}

\section{Achievability Scheme}\label{sec:IPLT-Ach}
This section presents an IPLT protocol, called \emph{Generalized Partition-and-Code with Partial Interference Alignment (GPC-PIA)}, that achieves the rate ${(\lfloor \frac{K}{D}\rfloor+\min\{\frac{R}{S},\frac{R}{L}\})^{-1}}$, where $R\triangleq K \pmod D$ and $S\triangleq \gcd(D+R,R)$. Examples of this protocol are given in Section~\ref{sec:EX}. 


In the description of the protocol, 
we denote by $\tilde{\mathrm{W}}$ a sequence of length $D$ (instead of a set of size $D$) obtained by randomly permuting the elements in the demand's support index set $\mathrm{W}$, and denote by $\tilde{\mathrm{V}}$ an $L\times D$ matrix obtained by applying the same permutation on the columns of the demand's coefficient matrix $\mathrm{V}$.

We consider two cases: (i) $L\leq S$, and (ii) $L>S$. In each case, the protocol consists of three steps.  

\vspace{0.125cm}
\textbf{Step 1:} The user constructs a matrix $\mathrm{G}$ and a permutation $\pi$, and sends them as the query $\mathrm{Q}^{[\mathrm{W},\mathrm{V}]}$ to the server. The construction of the matrix $\mathrm{G}$ and the permutation $\pi$ for the cases (i) and (ii) is as follows. 



\subsubsection*{{Case (i)}} Let $n\triangleq \lfloor \frac{K}{D}\rfloor-1$, $m\triangleq \frac{R}{S}+1$, and $t\triangleq \frac{D}{S}-1$. The user constructs an $L(n+m)\times K$ matrix $\mathrm{G}$,
\begin{equation}\label{eq:12}
\mathrm{G} = \begin{bmatrix} 
\mathrm{G}_1 & 0 & \dots & 0 & 0 \\
0 & \mathrm{G}_2 & \dots & 0 & 0\\
\vdots & \vdots & \ddots & \vdots & \vdots \\
0 & 0 & \dots & \mathrm{G}_{n} & 0\\ 
0 & 0 & \dots & 0 & \mathrm{G}_{n+1} 
\end{bmatrix}
\end{equation} where $\mathrm{G}_1,\dots\mathrm{G}_{n}$ are $L\times D$ matrices, and $\mathrm{G}_{n+1}$ is an ${Lm\times (D+R)}$ matrix. 
The blocks $\mathrm{G}_1,\dots,\mathrm{G}_n,\mathrm{G}_{n+1}$ are constructed as follows. 

The user randomly selects one of the blocks $\mathrm{G}_1,\dots,\mathrm{G}_{n+1}$, where each of the blocks $\mathrm{G}_1,\dots,\mathrm{G}_n$ is selected with probability $\frac{D}{K}$, and the block $\mathrm{G}_{n+1}$ is selected with probability $\frac{D+R}{K}$. 
Let $i^{*}$ be the index of the selected block. 
In the following, we consider the cases of $1\leq i^{*}\leq n$ and $i^{*}=n+1$ separately.  

First, consider the case of $1\leq i^{*}\leq n$. In this case, the user takes $\mathrm{G}_{i^{*}}$ to be the matrix $\tilde{\mathrm{V}}$, i.e., 
$\mathrm{G}_{i^{*}}=\tilde{\mathrm{V}}$. 
For any $i\in \{1,\dots,n\}\setminus \{i^{*}\}$, the user takes $\mathrm{G}_i$ to be a randomly generated MDS matrix of size ${L\times D}$. 
The construction of $\mathrm{G}_{n+1}$ is as follows.
First, the user randomly generates an MDS matrix $\mathrm{C}$ of size $L\times (D+R)$, and partitions the columns of $\mathrm{C}$ into $t+m$ column-blocks each of size $L\times S$, i.e., $\mathrm{C} = [\mathrm{C}_1,\dots,\mathrm{C}_{t+m}]$. 
Then, the user constructs ${\mathrm{G}_{n+1} = [\mathrm{B}_{1},\mathrm{B}_{2}]}$,  
\begin{equation*}
\mathrm{B}_{1} \triangleq 
\begin{bmatrix}
\alpha_1\omega_{1,1} \mathrm{C}_1 & \dots & \alpha_t\omega_{1,t} \mathrm{C}_t\\ 
\vdots & \vdots & \vdots \\ 
\alpha_1\omega_{m,1} \mathrm{C}_1 & \dots & \alpha_t\omega_{m,t} \mathrm{C}_t
\end{bmatrix},
\end{equation*}
\begin{equation*}
\mathrm{B}_{2} \triangleq 
\begin{bmatrix}
\alpha_{t+1} \mathrm{C}_{t+1} &   &  \\ 
 & \ddots &  \\
 &  & \alpha_{t+m} \mathrm{C}_{t+m} \\
\end{bmatrix},
\end{equation*} where 
$\alpha_1,\dots,\alpha_{t+m}$ are $t+m$ randomly chosen elements from $\mathbb{F}_p\setminus \{0\}$, and $\omega_{i,j}\triangleq (x_i-y_j)^{-1}$ for $1\leq i\leq m$ and $1\leq j\leq t$, where $x_1,\dots,x_m$ and $y_1,\dots,y_t$ are $t+m$ distinct elements chosen at random from $\mathbb{F}_p$. (Note that $\omega_{i,j}$ is the entry $(i,j)$ of an $m\times t$ Cauchy matrix.)


Now, consider the case of $i^{*}=n+1$. For any ${i\in \{1,\dots,n\}}$, the user takes $\mathrm{G}_i$ to be a randomly generated MDS matrix of size $L\times D$. 
The user then constructs $\mathrm{G}_{n+1}$ with a structure similar to that in the previous case, but for a different choice of matrices $\mathrm{C}_1,\dots,\mathrm{C}_{t+m}$ and parameters $\alpha_{1},\dots,\alpha_{t+m}$, as specified below. 

First, the user partitions the columns of $\tilde{\mathrm{V}}$ into $t+1$ column-blocks each of size $L\times S$, i.e., $\tilde{\mathrm{V}}= [\tilde{\mathrm{V}}_1,\dots,\tilde{\mathrm{V}}_{t+1}]$. 
The user then randomly chooses $t+1$ indices from ${\{1,\dots,t+m\}}$, say, $i_1,\dots,i_{t+1}$, and for any ${1\leq j\leq t+1}$, takes $\mathrm{C}_{i_j}=\tilde{\mathrm{V}}_j$. 
Next, the user randomly generates the rest of $\mathrm{C}_i$'s such that $\mathrm{C} = [\mathrm{C}_1,\dots,\mathrm{C}_{t+m}]$ is an MDS matrix. To choose $\alpha_i$'s, the user proceeds as follows. 

We refer to the submatrix of $\mathrm{G}_{n+1}$ formed by the $i$th $L$ rows as the $i$th row-block of $\mathrm{G}_{n+1}$. 
Note that $\mathrm{G}_{n+1}$ has $m$ row-blocks.
Let $s$ be the number of column-block indices $i_j$ for $j\in \{1,\dots,t+m\}$ such that $i_j>t$. 
Note that $\mathrm{C}_{i_1},\dots,\mathrm{C}_{i_{t-s+1}}$ belong to the matrix $\mathrm{B}_1$, and $\mathrm{C}_{i_{t-s+2}},\dots,\mathrm{C}_{i_{t+1}}$ belong to the matrix $\mathrm{B}_2$. 
Let $\mathcal{I}\triangleq\{i_1,\dots,i_{t+1}\}$ be the index set of those column-blocks of $\mathrm{C}$ that correspond to the column-blocks of $\tilde{\mathrm{V}}$. 
Let $\mathcal{I}_1\triangleq \{i_1,\dots,i_{t-s+1}\}$, and let $\mathcal{I}_2\triangleq \mathcal{I}\setminus \mathcal{I}_1$. 
Note that for any $i\in \mathcal{I}_1$, $\mathrm{C}_{i}$ appears in all row-blocks of $\mathrm{G}_{n+1}$, and 
for any $i\in \mathcal{I}_2$, $\mathrm{C}_{i}$ appears only in the $(i-t)$th row-block of $\mathrm{G}_{n+1}$. 

The parameters $\alpha_i$'s are to be chosen such that, by performing row-block operations on $\mathrm{G}_{n+1}$, the user can construct an $L\times (D+R)$ matrix with $t+m$ column-blocks each of size $L\times S$ that satisfies the following two conditions: 
(a) the blocks indexed by ${\{1,\dots,t+m\}\setminus \mathcal{I}}$ are all zero, and 
(b) the blocks indexed by $\mathcal{I} = \{i_1,\dots,i_{t+1}\}$ are $\mathrm{C}_{i_1},\dots,\mathrm{C}_{i_{t+1}}$. 
For simplifying the notation, let 
$\{j_1,\dots,j_{s-1}\}\triangleq {\{1,\dots,t\}\setminus \mathcal{I}}$, and let $\{k_1,\dots,k_s\}\triangleq \mathcal{I}_2 =  \{i_{t-s+2},\dots,i_{t+1}\}$. 

To perform row-block operations, for every $i\in \mathcal{I}_2 = \{k_1,\dots,k_s\}$, 
the user multiplies the $(i-t)$th row-block of $\mathrm{G}_{n+1}$ by a nonzero coefficient $c_{i}$. 
Let $\mathrm{c} \triangleq [c_{k_1},\dots,c_{k_s}]^{\mathsf{T}}$. 
Upon choosing $\alpha_{j_1},\dots,\alpha_{j_{s-1}}$ randomly from ${\mathbb{F}_p\setminus \{0\}}$, it follows that the condition (a) is satisfied so long as $\mathrm{M}_1\mathrm{c}$ is an all-zero vector, where  
\begin{equation*}
\mathrm{M}_1 \triangleq 
\begin{bmatrix}
\omega_{k_1-t,j_1} & \omega_{k_2-t,j_1} & \dots & \omega_{k_s-t,j_1}\\
\omega_{k_1-t,j_2} & \omega_{k_2-t,j_2} & \dots & \omega_{k_s-t,j_2}\\
\vdots & \vdots & \vdots & \vdots \\
\omega_{k_1-t,j_{s-1}} & \omega_{k_2-t,j_{s-1}} & \dots & \omega_{k_s-t,j_{s-1}}
\end{bmatrix}.
\end{equation*} Since $\mathrm{M}_1$ is a Cauchy matrix by the choice of $\omega_{i,j}$'s, the submatrix of $\mathrm{M}_1$ formed by columns indexed by $\{2,\dots,s\}$ (and all $s-1$ rows) is invertible~\cite{R2006}. 
Thus, for any arbitrary ${c_{k_1}\neq 0}$, there exists a unique solution for the vector $\mathrm{c}$. 
Note also that all the components of $\mathrm{c}$ are nonzero because $\mathrm{M}_{1}$ is a super-regular matrix, i.e., every square submatrix of $\mathrm{M}_{1}$ is invertible (by the properties of Cauchy matrices). 

Given the vector $\mathrm{c}$, the condition (b) is satisfied so long as ${\alpha_{k_1} = 1/c_{k_1}}, \dots, {\alpha_{k_s} = 1/c_{k_s}}$, and 
$\alpha_{i_1},\dots,\alpha_{i_{t-s+1}}$ are chosen such that $\mathrm{M}_2\mathrm{c}$ is an all-one vector, where  
\begin{equation*}
\mathrm{M}_2 \triangleq 
\begin{bmatrix}
\alpha_{i_1}\omega_{k_1-t,i_1} & \dots & \alpha_{i_1}\omega_{k_s-t,i_1}\\
\alpha_{i_2}\omega_{k_1-t,i_2} & \dots & \alpha_{i_2}\omega_{k_s-t,i_2}\\
\vdots & \vdots & \vdots \\
\alpha_{i_{t-s+1}}\omega_{k_1-t,i_{t-s+1}} & \dots & \alpha_{i_{t-s+1}}\omega_{k_s-t,i_{t-s+1}}
\end{bmatrix}.
\end{equation*} 
Solving for $\alpha_{i_1},\dots,\alpha_{i_{t-s+1}}$, it follows that $\alpha_{i_j} \triangleq (\sum_{1\leq l\leq s} c_{k_l}\omega_{k_l-t,i_j})^{-1}$ for ${1\leq j\leq t-s+1}$. 
Note that $\alpha_{i_1},\dots,\alpha_{i_{t-s+1}}$ are nonzero. 
It should also be noted that $\sum_{1\leq l\leq s}c_{k_l}\omega_{k_l-t,i_j}$ is nonzero because the $j$th row of $\mathrm{M}_2$ 
is linearly independent of the rows of $\mathrm{M}_1$. 
Lastly, for any $i\in \{1,\dots,t+m\}\setminus \{i_1,\dots,i_{t+1},j_1,\dots,j_{s-1}\}$, the user chooses $\alpha_i$ randomly from $\mathbb{F}_p\setminus \{0\}$. 
This completes the construction of the matrix $\mathrm{G}$. 


Next, the user constructs a permutation $\pi$ as follows.  
Let $\tilde{\mathrm{W}}=\{l_1,\dots,l_{D}\}$, and let $\mathcal{K}\setminus \mathrm{W}=\{l_{D+1},\dots,l_{K}\}$. 
First, consider the case of $1\leq i^{*}\leq n$. 
In this case, the user constructs a permutation $\pi$ such that: for every ${1\leq j\leq D}$, $\pi(l_j) = (i^*-1)D+j$; and for every $D<j\leq K$, $\pi(l_j)$ is a randomly chosen element from $\mathcal{K}\setminus \{\pi(l_k)\}_{1\leq k<j}$. 
Next, consider the case of $i^{*} = n+1$. Recall that $i_1,\dots,i_{t+1}$ are the indices of those column-blocks of $\mathrm{C}$ that correspond to the column-blocks of $\tilde{\mathrm{V}}$. 
The user constructs a permutation $\pi$ such that: 
for every $1\leq k\leq {t+1}$ and $(k-1)S+1\leq j\leq kS$, ${\pi(l_j) = nD+(i_{k}-1)S+f_j}$, where $f_j= j\pmod S$ if $S\nmid j$, and $f_j = S$ if $S\mid j$; and for every $D<j\leq K$, $\pi(l_j)$ is a randomly chosen element from $\mathcal{K}\setminus \{\pi(l_k)\}_{1\leq k<j}$. 

\vspace{0.125cm}
\subsubsection*{{Case (ii)}} Let $n\triangleq \lfloor \frac{K}{D}\rfloor-1$, and $m\triangleq \frac{R}{L}+1$. The user constructs an $L(n+m)\times K$ matrix $\mathrm{G}$ with a structure similar to~\eqref{eq:12}, where $\mathrm{G}_1,\dots,\mathrm{G}_n$ are constructed similarly as in the previous case, but the construction of $\mathrm{G}_{n+1}$ is different. Below, we will only explain how to construct $\mathrm{G}_{n+1}$. 

For the case of $1\leq i^{*}\leq n$, the user randomly generates an ${[D+R,L+R]}$ MDS code, and takes $\mathrm{G}_{n+1}$ to be the generator matrix of this code. 
For the case of $i^{*}=n+1$, the user constructs an ${[D+R,L+R]}$ MDS code using the same technique as in the step 1 of the Specialized MDS Code protocol of~\cite{EHS2021Joint}, 
except where $K$ is replaced by $D+R$, and $\mathrm{W}$ is replaced by a randomly chosen $D$-subset of $\{1,\dots,D+R\}$, say, $\{h_1,\dots,h_D\}$. 
The user then uses the generator matrix of the constructed MDS code as $\mathrm{G}_{n+1}$. 


Next, the user constructs a permutation $\pi$.
For the case of ${1\leq i^{*}\leq n}$, the permutation $\pi$ is generated exactly the same as in the case (i), whereas the construction of the permutation $\pi$ for the case of $i^{*}=n+1$ is different from that in the case (i). 
Similarly as before, let $\tilde{\mathrm{W}}=\{l_1,\dots,l_{D}\}$, and let $\mathcal{K}\setminus \mathrm{W}=\{l_{D+1},\dots,l_{K}\}$. 
For the case of $i^{*}=n+1$, the user constructs a permutation $\pi$ such that: for every ${1\leq j\leq D}$, $\pi(l_j) = nD+h_{j}$; and for every ${D<j\leq K}$, $\pi(l_j)$ is a randomly chosen element from $\mathcal{K}\setminus \{\pi(l_k)\}_{1\leq k<j}$.

\vspace{0.125cm}
\textbf{Step 2:} Given the query $\mathrm{Q}^{[\mathrm{W},\mathrm{V}]}$, i.e., the matrix $\mathrm{G}$ and the permutation $\pi$, the server first constructs the vector $\tilde{\mathrm{X}} \triangleq \pi(\mathrm{X})$ by permuting the components of the vector $\mathrm{X}$ according to the permutation $\pi$, i.e.,  $\tilde{\mathrm{X}}_{\pi(l)} \triangleq \mathrm{X}_l$ for $l\in \mathcal{K}$. Then, the server  
computes the vector $\mathrm{y}\triangleq \mathrm{G}\tilde{\mathrm{X}}$, and sends $\mathrm{y}$ back to the user as the answer $\mathrm{A}^{[\mathrm{W},\mathrm{V}]}$.

\vspace{0.125cm}
\textbf{Step 3:} Upon receiving the answer $\mathrm{A}^{[\mathrm{W},\mathrm{V}]}$, i.e., the vector $\mathrm{y}$, the user recovers the demand vector $\mathrm{Z}^{[\mathrm{W},\mathrm{V}]}$ as follows. 
For every $1\leq i\leq n$, let $\mathrm{y}_i$ be the vector $\mathrm{y}$ restricted to its components indexed by $\{(i-1)L+1,\dots,iL\}$, and let $\mathrm{y}_{n+1}$ be the vector $\mathrm{y}$ restricted to its components indexed by $\{nL+1,\dots,nL+mL\}$. 
For the case of $1\leq i^{*}\leq n$, the demand $\mathrm{Z}^{[\mathrm{W},\mathrm{V}]}$ can be recovered from the vector $\mathrm{y}_{i^{*}}$ for both cases (i) and (ii). 
For the case of ${i^{*}=n+1}$, the demand $\mathrm{Z}^{[\mathrm{W},\mathrm{V}]}$ can be recovered 
by performing proper row-block or row operations on the augmented matrix $[\mathrm{G}_{n+1},\mathrm{y}_{n+1}]$ for the case (i) or (ii), respectively. 

\begin{lemma}\label{lem:IPLT-Ach}
The GPC-PIA protocol is an IPLT protocol, and achieves the rate ${(\lfloor \frac{K}{D}\rfloor+\min\{\frac{R}{S},\frac{R}{L}\})^{-1}}$, where ${R\triangleq K \pmod D}$ and $S\triangleq \gcd(D+R,R)$.
\end{lemma}


\begin{proof}
To avoid repetition, we only present the proof for the case (i). Using similar arguments, the results can be proven for the case (ii).  

In the case (i), it is easy to see that the rate of the protocol is ${L\theta/(L(n+m)\theta)} = {(n+m)^{-1}} = (\lfloor \frac{K}{D}\rfloor+\frac{R}{S})^{-1}$. This is because the matrix $\mathrm{G}$ has $L(n+m)$ rows, and the vector $\mathrm{y}=\mathrm{G}\tilde{\mathrm{X}}$ contains $L(n+m)$ independently and uniformly distributed components, each with entropy $\theta$. From the construction, it should also be obvious that the recoverability condition is satisfied. 

Next, we show that the individual privacy condition is satisfied. 
Let $\tilde{\mathrm{X}} \triangleq [X_{i_1},\dots,X_{i_K}]^{\mathsf{T}}$. 
For every ${1\leq j\leq n}$, let $\mathcal{I}_j$ be the set of $j$th group of $D$ elements in $\{i_1,\dots,i_{nD}\}$, and for every $1\leq j\leq t+m$, let $\mathcal{I}_{n+j}$ be the set of $j$th group of $S$ elements in $\{i_{nD+1},\dots,i_K\}$. For any positive integers $a,b$ such that $b\leq a$, we denote by $C_{a,b}$ the binomial coefficient $\binom{a}{b}$.
For every $1\leq j\leq n$, let $\mathrm{W}_j\triangleq \mathcal{I}_j$, and for every $1\leq j\leq r\triangleq C_{t+m,t+1}$, let $\mathrm{W}_{n+j} = \cup_{k\in \mathcal{J}_j} \mathcal{I}_{k}$, where $\mathcal{J}_1,\dots,\mathcal{J}_{r}$ are all $(t+1)$-subsets of ${\{n+1,\dots,n+t+m\}}$.
Note that $\mathrm{W}_1,\dots,\mathrm{W}_n,\mathrm{W}_{n+1},\dots,\mathrm{W}_{n+t+m}$ are the only possible demand's support index sets from the perspective of the server, given the user's query. 

For the ease of notation, let ${\mathrm{Q} \triangleq \{\mathrm{G},\pi\}}$ be the user's query. 
To prove that the individual privacy condition is satisfied, we need to show that ${\Pr(i\in \mathbf{W}|\mathbf{Q}=\mathrm{Q})} = {\Pr(i\in \mathbf{W})}$ for every $i\in \mathcal{K}$, or equivalently, ${\Pr(i\in \mathbf{W}|\mathbf{Q}=\mathrm{Q})}$ is the same for all $i\in \mathcal{K}$. 
Consider an arbitrary $i\in \mathcal{K}$. There are two different cases: (i) ${\pi(i)\leq nD}$, and (ii) ${\pi(i)>nD}$. 

First, consider the case (i). In this case, there exists a unique $j$ (for any ${1\leq j\leq n}$) such that $i\in \mathrm{W}_j$. Thus, ${\Pr(i\in \mathbf{W}|\mathbf{Q}=\mathrm{Q})} = {\Pr(\mathbf{W}=\mathrm{W}_j|\mathbf{Q}=\mathrm{Q})}$. 
By applying Bayes' rule, we have
\begin{align}
& \Pr(\mathbf{W}=\mathrm{W}_j|\mathbf{Q}=\mathrm{Q}) \nonumber \\ 
& = \frac{\Pr(\mathbf{Q}=\mathrm{Q}|\mathbf{W}=\mathrm{W}_j)}{\Pr(\mathbf{Q}=\mathrm{Q})}\Pr(\mathbf{W}=\mathrm{W}_j)\nonumber \\
& = \frac{\Pr(\mathbf{Q}=\mathrm{Q}|\mathbf{W}=\mathrm{W}_j)}{\Pr(\mathbf{Q}=\mathrm{Q})}\times \frac{1}{C_{K,D}}.\label{eq:13}
\end{align}
The structure of $\mathrm{G}$---the size and the position of the blocks $\mathrm{G}_1,\dots,\mathrm{G}_{n+1}$---does not depend on $(\mathrm{W},{\pi})$, and the matrix $\mathrm{V}$ and all other MDS matrices used in the construction of $\mathrm{G}$ are generated independently from $(\mathrm{W},{\pi})$. 
This implies that $\mathbf{G}$ is independent of $(\mathbf{W},\boldsymbol{\pi})$. Thus, $\Pr(\mathbf{Q}=\mathrm{Q}) = {\Pr(\mathbf{G}=\mathrm{G},\boldsymbol{\pi}=\pi)}=\Pr(\mathbf{G}=\mathrm{G})\Pr(\boldsymbol{\pi}=\pi)$, and 
\begin{align}
& \frac{\Pr(\mathbf{Q}=\mathrm{Q}|\mathbf{W}=\mathrm{W}_j)}{\Pr(\mathbf{Q}=\mathrm{Q})} \nonumber \\
& = \frac{\Pr(\mathbf{G}=\mathrm{G},\boldsymbol{\pi}=\pi|\mathbf{W}=\mathrm{W}_j)}{\Pr(\mathbf{G}=\mathrm{G})\Pr(\boldsymbol{\pi}=\pi)} \nonumber \\ 
& = \frac{\Pr(\mathbf{G} = \mathrm{G})\Pr(\boldsymbol{\pi}=\pi|\mathbf{W}=\mathrm{W}_j)}{\Pr(\mathbf{G}=\mathrm{G})\Pr(\boldsymbol{\pi}=\pi)} \nonumber \\ 
& = \frac{\Pr(\boldsymbol{\pi}=\pi|\mathbf{W}=\mathrm{W}_j)}{\Pr(\boldsymbol{\pi}=\pi)} \label{eq:14}
\end{align} 
Given $\mathbf{W} = \mathrm{W}_j$, the conditional probability of the event of $\boldsymbol{\pi}=\pi$ is equal to the joint probability of the events of ${\boldsymbol{\pi}(\mathbf{W}) = \pi(\mathrm{W}_j)}$ and ${\boldsymbol{\pi}(\mathcal{K}\setminus \mathbf{W}) = \pi(\mathcal{K}\setminus\mathrm{W}_j)}$. 
Note that $\Pr(\boldsymbol{\pi}(\mathbf{W}) = \pi(\mathrm{W}_j)) = \Pr(\boldsymbol{i^{*}} = j)\times \frac{1}{D!}=\frac{D}{K}\times \frac{1}{D!}$, where $i^{*}$ is the index of the block selected in the step 1 of the protocol, and 
${\Pr(\boldsymbol{\pi}(\mathcal{K}\setminus \mathbf{W}) = \pi(\mathcal{K}\setminus \mathrm{W}_j))} = \frac{1}{(K-D)!}$ by the construction of the permutation $\pi$ in the step 1 of the protocol. Thus, 
\begin{equation}\label{eq:15}
\Pr(\boldsymbol{\pi}=\pi|\mathbf{W}=\mathrm{W}_j) = \frac{D}{K}\times\frac{1}{D!}\times \frac{1}{(K-D)!}.    
\end{equation}
Combining~\eqref{eq:13}-\eqref{eq:15}, we have
\begin{equation}\label{eq:16}
\Pr(i\in \mathbf{W}|\mathbf{Q}=\mathrm{Q}) =
\frac{1}{\Pr(\boldsymbol{\pi}=\pi)}\times\frac{D}{K}\times\frac{1}{K!}.
\end{equation}

Now, consider the case (ii). In this case, there exist $s\triangleq C_{t+m-1,t}$ distinct indices $j_1,\dots,j_{s}$ (${1\leq j_1,\dots,j_s\leq r}$) such that ${i\in \mathrm{W}_{n+j_1}, \dots, i\in \mathrm{W}_{n+j_s}}$. 
Using similar arguments as those in the case (i), we have
\begin{align}
& {\Pr(i\in \mathbf{W}|\mathbf{Q}=\mathrm{Q})} \nonumber\\
& = {\sum_{1\leq k\leq s}\Pr(\mathbf{W}=\mathrm{W}_{n+j_k}|\mathbf{Q}=\mathrm{Q})}\nonumber\\
&  = \sum_{1\leq k\leq s} \frac{\Pr(\mathbf{Q}=\mathrm{Q}|\mathbf{W}=\mathrm{W}_{n+j_k})}{\Pr(\mathbf{Q}=\mathrm{Q})}\Pr(\mathbf{W}=\mathrm{W}_{n+j_k}) \nonumber\\
&  = \sum_{1\leq k\leq s} \frac{\Pr(\mathbf{G} = \mathrm{G})\Pr(\boldsymbol{\pi}=\pi|\mathbf{W}=\mathrm{W}_{n+j_k})}{\Pr(\mathbf{G}=\mathrm{G})\Pr(\boldsymbol{\pi}=\pi)}\times\frac{1}{C_{K,D}}\nonumber\\
&  = \frac{1}{\Pr(\boldsymbol{\pi}=\pi)} \sum_{1\leq k\leq s} \Pr(\boldsymbol{\pi}=\pi|\mathbf{W}=\mathrm{W}_{n+j_k})\times\frac{1}{C_{K,D}}\nonumber\\
&  = \frac{1}{\Pr(\boldsymbol{\pi}=\pi)} \sum_{1\leq k\leq s} \left(\frac{D+R}{K}\times \frac{1}{r}\times\frac{1}{D!}\times\frac{1}{(K-D)!}\right)\frac{1}{C_{K,D}}\nonumber\\
&  = \frac{1}{\Pr(\boldsymbol{\pi}=\pi)} \times s \left(\frac{D+R}{K}\times \frac{1}{r}\times\frac{1}{D!}\times\frac{1}{(K-D)!}\right)\frac{1}{C_{K,D}}\nonumber\\
& = \frac{1}{\Pr(\boldsymbol{\pi}=\pi)}\times \frac{s}{r}\times \frac{D+R}{K}\times \frac{1}{K!}\nonumber\\
& = \frac{1}{\Pr(\boldsymbol{\pi}=\pi)}\times \frac{D}{D+R}\times \frac{D+R}{K}\times \frac{1}{K!}\nonumber\\
& = \frac{1}{\Pr(\boldsymbol{\pi}=\pi)}\times \frac{D}{K}\times \frac{1}{K!}.\label{eq:17}
\end{align} 
Comparing~\eqref{eq:16} and~\eqref{eq:17}, it follows that ${\Pr(i\in \mathbf{W}|\mathbf{Q}=\mathrm{Q})}$ is the same for all $i\in \mathcal{K}$, as was to be shown. 
\end{proof}


\section{Examples of the GPC-PIA Protocol}\label{sec:EX}
In this section, we provide two illustrative examples of the GPC-PIA protocol. Example~1 corresponds to a scenario with $L\leq S$, and Example~2 corresponds to a scenario with $L>S$. 
\begin{example}
\normalfont
Consider  a  scenario  where  the  server has $K=20$ messages, ${X}_1,\dots,{X}_{20}\in\mathbb{F}_{13}$, and the user wishes to compute $L=3$ linear combinations of $D=8$ messages ${X}_2,{X}_4,{X}_5,{X}_7,{X}_8,{X}_{10},{X}_{11},{X}_{12}$, say, 
\begin{align*}
Z_{1} & =7{X}_2+3{X}_4+12{X}_5+10{X}_7+2{X}_8+{X}_{10}+5{X}_{11}+6{X}_{12}	\\
Z_{2} & =3{X}_2+6{X}_4+5{X}_5+12{X}_7+8{X}_8+3{X}_{10}+11{X}_{11}+4{X}_{12} \\
Z_{3} & =5{X}_2+12{X}_4+{X}_5+4{X}_7+6{X}_8+9{X}_{10}+6{X}_{11}+7{X}_{12}
\end{align*}
For this example, the demand's support index set $\mathrm{W}$ is given by $\mathrm{W}=\{2,4,5,7,8,10,11,12\}$, and the demand's coefficient matrix $\mathrm{V}$ is given by 
\begin{equation*}
\mathrm{V} = 
\begin{bmatrix}
7 & 3 & 12 & 10 & 2 & 1 & 5 & 6\\
3 & 6 & 5 & 12 & 8 & 3 & 11 & 4\\
5 & 12 & 1 & 4 & 6 & 9 & 6 & 7\\
\end{bmatrix}.
\end{equation*} It is easy to verify that the matrix $\mathrm{V}$ is MDS, i.e., every $3\times 3$ submatrix of $\mathrm{V}$ is invertible (over $\mathbb{F}_{13}$). 

Let $\tilde{\mathrm{W}}$ be a sequence of length $8$ obtained by randomly permuting the elements in $\mathrm{W}$, for example, $\tilde{\mathrm{W}}=\{10,4,8,2,7,5,11,12\}$, 
and let $\tilde{\mathrm{V}}$ be a $3\times 8$ matrix obtained by applying the same permutation on the columns of the matrix $\mathrm{V}$, i.e., 
\begin{equation*}
\tilde{\mathrm{V}} = 
\begin{bmatrix}
1 & 3 & 2 & 7 & 10 & 12 & 5 & 6\\
3 & 6 & 8 & 3 & 12 & 5 & 11 & 4\\
9 & 12 & 6 & 5 & 4 & 1 & 6 & 7\\
\end{bmatrix},
\end{equation*} 
Note that $\mathrm{V}\mathrm{X}_{\mathrm{W}} = \tilde{\mathrm{V}}\mathrm{X}_{\tilde{\mathrm{W}}}$ by the construction of $\tilde{\mathrm{W}}$ and $\tilde{\mathrm{V}}$. 

For this example, $R= K \pmod D =4$, $S=\gcd(D+R,R)=4$, $n=\lfloor \frac{K}{D}\rfloor-1=1$, $m=\frac{R}{S}+1=2$, and $t=\frac{D}{S}-1=1$. Note that for this example, $L=3<S=4$. 

The query of the user consists of a $9\times20$ matrix $\mathrm{G}$ and a permutation $\pi$, constructed as follows. 
The matrix $\mathrm{G}$ consists of two submatrices (blocks) $\mathrm{G}_1$ and $\mathrm{G}_2$ of size $3\times8$ and $6\times12$, respectively, i.e., 
\begin{equation}\label{eq:14}
\mathrm{G} = 
\begin{bmatrix}
\mathrm{G}_1 & 0_{3\times 12}\\
0_{6\times 8} & \mathrm{G}_2
\end{bmatrix},
\end{equation} where the construction of the blocks $\mathrm{G}_1$ and $\mathrm{G}_2$ is described below.

The user randomly selects one of the blocks $\mathrm{G}_1,\mathrm{G}_2$, where the probability of selecting the block $\mathrm{G}_1$ is $\frac{8}{20}$, and the probability of selecting the block $\mathrm{G}_2$ is $\frac{12}{20}$. 
Depending on $\mathrm{G}_1$ or $\mathrm{G}_2$ being selected, the construction of each of these blocks is different. In this example, we consider the case that the user selects $\mathrm{G}_2$. 
In this case, the user takes $\mathrm{G}_1$ to be a randomly generated MDS matrix of size $3\times 8$, say, 
\begin{equation}\label{eq:15}
\mathrm{G}_{1} = \begin{bmatrix}
    5 & 8 & 4 & 7 & 4 & 3 & 4 & 2\\
    7 & 4 & 12 & 9 & 1 & 10 & 6 & 5\\
    2 & 2 & 10 & 6 & 10 & 3 & 9 & 6
\end{bmatrix}.
\end{equation}
To construct $\mathrm{G}_{2}$, the user first constructs a $3\times 12$ matrix $\mathrm{C} = [\mathrm{C}_1,\mathrm{C}_2,\mathrm{C}_{3}]$, where the column-blocks $\mathrm{C}_1,\mathrm{C}_2,\mathrm{C}_3$, each of size $3\times 4$, are 
constructed as follows. The user partitions the columns of $\tilde{\mathrm{V}}$ into two column-blocks $\tilde{\mathrm{V}}_1$ and $\tilde{\mathrm{V}}_2$, each of size $3\times 4$, i.e., 
\[
\tilde{\mathrm{V}}_{1} = 
\begin{bmatrix}
    1 & 3 & 2 & 7\\
    3 & 6 & 8 & 3\\
    9 & 12 & 6 & 5
\end{bmatrix}, \quad \quad 
\tilde{\mathrm{V}}_{2} = 
\begin{bmatrix}
    10 & 12 & 5 & 6\\
    12 & 5 & 11 & 4\\
    4 & 1 & 6 & 7
\end{bmatrix}.
\]
The user then randomly chooses two indices $i_1,i_2$ from $\{1,2,3\}$, say, $i_1=2$ and $i_2=3$, and takes $\mathrm{C}_{i_1}=\mathrm{C}_2 = \tilde{\mathrm{V}}_1$, and $\mathrm{C}_{i_2} = \mathrm{C}_{3}=\tilde{\mathrm{V}}_2$.
Next, the user takes the remaining column-block of $\mathrm{C}$, i.e., $\mathrm{C}_1$, to be a randomly generated matrix of size $3\times 4$ such that 
$\mathrm{C} = [\mathrm{C}_1,\mathrm{C}_2,\mathrm{C}_{3}]$ is an MDS matrix. For this example, suppose $\mathrm{C}_1$ is 
given by
\begin{equation*}
\mathrm{C}_{1}=
\begin{bmatrix}
1 & 9 & 11 & 2\\
7 & 9 & 2 & 9\\
10 & 9 & 11 & 8
\end{bmatrix}.
\end{equation*} It is easy to verify that the matrix $\mathrm{C}$ is MDS. Next, the user constructs $\mathrm{G}_2$ as 
\[
\mathrm{G}_{2} = 
\begin{bmatrix}
\alpha_1\mathrm{C}_1 & \alpha_{2}\mathrm{C}_{2} & 0_{3\times 4}\\
\alpha_1\mathrm{C}_1 & 0_{3\times 4} & \alpha_{3}\mathrm{C}_{3}
\end{bmatrix},
\] where the (scalar) parameters $\alpha_1,\alpha_2,\alpha_3$ are chosen  
so that, by performing row-block operations on $\mathrm{G}_{2}$, the user can obtain the matrix $[0_{3\times 4},\mathrm{C}_2,\mathrm{C}_3]$. To do so, 
the user randomly chooses $\alpha_{1}$ form $\mathbb{F}_{13}\setminus{\{0\}}$, say $\alpha_{1}=2$. Note that 
the first column-block of $\mathrm{G}_{2}$, i.e., $\alpha_1 \mathrm{C}_1$, does not contain any column-blocks of $\tilde{\mathrm{V}}$, and hence must be eliminated by row-block operations. 
To perform row-block operations, the user multiplies the first row-block $[\alpha_1\mathrm{C}_1,\alpha_2\mathrm{C}_2,0_{3\times 4}]$ by a scalar $c_2$ and the second row-block $[\alpha_1\mathrm{C}_1,0_{3\times 4},\alpha_3\mathrm{C}_3]$ by a scalar $c_3$, and obtains $[(c_2+c_3)\alpha_1\mathrm{C}_1,c_2\alpha_2\mathrm{C}_2,c_3\alpha_3\mathrm{C}_3]$. 
Thus, the user finds $\alpha_2,\alpha_3$ and $c_3$ such that $c_2+c_3=0$, $c_2\alpha_2=1$, and $c_3\alpha_3=1$ for an arbitrary choice of $c_2\neq 0$, say, $c_2 = 4$.  
Then, $c_3 = -c_2 = 9$, $\alpha_2 = \frac{1}{c_2} = 10$, and $\alpha_3 = \frac{1}{c_3} = 3$.   The user then constructs $\mathrm{G}_2$ as 
\begin{equation}\label{eq:16}
\mathrm{G}_{2} = 
\begin{bmatrix}
2\mathrm{C}_1 & 10\mathrm{C}_{2} & 0_{3\times 4}\\
2\mathrm{C}_1 & 0_{3\times 4} & 3\mathrm{C}_{3}
\end{bmatrix} 
= \begin{bmatrix}
2 & 5 & 9 & 4 & 10 & 4 & 7  & 5 & 0 & 0 & 0 & 0\\
1 & 5 & 4 & 5 & 4 & 8 & 2  & 4 & 0 & 0 & 0 & 0\\	
7 & 5 & 9 & 3 & 12 & 3 & 8  & 11 & 0 & 0 & 0 & 0\\	
2 & 5 & 9 & 4 & 0 & 0 & 0  & 0 & 4 & 10 & 2 & 5\\
1 & 5 & 4 & 5 & 0 & 0 & 0  & 0 & 10 & 2 & 7 & 12\\	
7 & 5 & 9 & 3 & 0 & 0 & 0  & 0 & 12 & 3 & 5 & 8
 \end{bmatrix}.
\end{equation} Combining $\mathrm{G}_1$ and $\mathrm{G}_2$ given by~\eqref{eq:15} and~\eqref{eq:16}, the user then constructs $\mathrm{G} $ as in~\eqref{eq:14}.

Next, the user constructs a permutation $\pi$ on $\{1,\dots,20\}$, using the procedure described in the step 1 of the protocol, and sends the permutation $\pi$ together with the matrix $\mathrm{G}$ to the server. 
Note that the columns $13,\dots,20$ of $\mathrm{G}$ are constructed based on the columns of $\tilde{\mathrm{V}}$, and the columns of $\tilde{\mathrm{V}}$ correspond to the messages $X_{10},X_4,X_8,X_2,X_7,X_5,X_{11},X_{12}$, respectively. 
Thus, the user constructs the permutation $\pi$ such that 
$\{\pi(10),\pi(4),\pi(8),\pi(2),\pi(7),\pi(5),\pi(11),\pi(12)\}=\{13,\dots,20\}$. 
For $i\in \{1,\dots,20\}\setminus \{2,4,5,7,8,10,11,12\}$, the user then randomly chooses $\pi(i)$ (subject to the constraint that $\pi$ forms a valid permutation on $\{1,\dots,20\}$). For this example, suppose the user takes $\{\pi(1),\pi(3),\pi(6),\pi(9),\pi(13),\dots,\pi(20)\} = \{3,8,1,2,9,11,5,4,12,6,10,7\}$. 




Upon receiving the user's query, i.e., the matrix $\mathrm{G}$ and the permutation $\pi$, the server first permutes the components of the vector $\mathrm{X} = [X_1,\dots,X_{20}]^{\mathsf{T}}$ according to the permutation $\pi$ to obtain the vector $\tilde{\mathrm{X}}=\pi(\mathrm{X})$, i.e., $\tilde{X}_{\pi(i)} = X_i$ for $i\in \{1,\dots,20\}$. For this example, the vector $\tilde{\mathrm{X}}$ is given by
\[\tilde{\mathrm{X}} = [X_6,X_9,X_1,X_{16},X_{15},X_{18},X_{20},X_3,X_{13},X_{19},X_{14},X_{17},X_{10},X_4,X_8,X_2,X_{7},X_{5},X_{11},X_{12}]^{\mathsf{T}}.\]
Then the server computes $\mathrm{y=\mathrm{G}\mathrm{\tilde{\mathrm{X}}}}$, and sends the vector $\mathrm{y}$ back to the user as the answer. 
Let $\mathrm{T}_1 = \{1,\dots,8\}$, $\mathrm{T}_2 = \{9,10,11,12\}$, $\mathrm{T}_3 = \{13,14,15,16\}$, and $\mathrm{T}_4 = \{17,18,19,20\}$. 
For any $T\subset \{1,\dots,20\}$, we denote by $\tilde{\mathrm{X}}_{\mathrm{T}}$ the vector $\tilde{\mathrm{X}}$ restricted to its components indexed by $\mathrm{T}$. 
Note that $[\tilde{\mathrm{X}}_{\mathrm{T}_3}^{\mathsf{T}},\tilde{\mathrm{X}}_{\mathrm{T}_4}^{\mathsf{T}}]^{\mathsf{T}} = \mathrm{X}_{\tilde{\mathrm{W}}}$, and $\mathrm{y} = [\mathrm{y}_1^{\mathsf{T}},\mathrm{y}_2^{\mathsf{T}}]^{\mathsf{T}}$, where $\mathrm{y}_1 \triangleq \mathrm{G}_1 \tilde{\mathrm{X}}_{\mathrm{T}_1}$, and $\mathrm{y}_2 \triangleq \mathrm{G}_2 [\tilde{\mathrm{X}}_{\mathrm{T}_2}^{\mathsf{T}},\tilde{\mathrm{X}}_{\mathrm{T}_3}^{\mathsf{T}},\tilde{\mathrm{X}}_{\mathrm{T}_4}^{\mathsf{T}}]^{\mathsf{T}}$. 
Let $\mathrm{I}$ be the identity matrix of size $3\times 3$. 
Then, the user recovers $[Z_1,Z_2,Z_3]^{\mathsf{T}} = \mathrm{V}\mathrm{X}_{\mathrm{W}} = \tilde{\mathrm{V}}\mathrm{X}_{\tilde{\mathrm{W}}}$ by computing 
\begin{align*}
& \begin{bmatrix} c_2\mathrm{I} & c_3\mathrm{I} \end{bmatrix}\mathrm{y}_2 = \begin{bmatrix} c_2\mathrm{I} & c_3\mathrm{I} \end{bmatrix} \mathrm{\mathrm{G}}_{2} \begin{bmatrix} \tilde{\mathrm{X}}_{\mathrm{T}_2}\\ \tilde{\mathrm{X}}_{\mathrm{T}_3} \\ \tilde{\mathrm{X}}_{\mathrm{T}_4}\end{bmatrix} 
= \begin{bmatrix} c_2\mathrm{I} & c_3\mathrm{I} \end{bmatrix} 
\begin{bmatrix}
2\mathrm{C}_1 & 10\mathrm{C}_{2} & 0\\
2\mathrm{C}_1 & 0 & 3\mathrm{C}_{3}
\end{bmatrix} \begin{bmatrix} \tilde{\mathrm{X}}_{\mathrm{T}_2}\\ \tilde{\mathrm{X}}_{\mathrm{T}_3}\\ \tilde{\mathrm{X}}_{\mathrm{T}_4}\end{bmatrix}\\
& \quad =  
\begin{bmatrix} 
2(c_2+c_3) \mathrm{C}_1 & 10c_2 \mathrm{C}_2 & 3c_3\mathrm{C}_3 
\end{bmatrix} 
\begin{bmatrix} \tilde{\mathrm{X}}_{\mathrm{T}_2}\\ \tilde{\mathrm{X}}_{\mathrm{T}_3}\\ \tilde{\mathrm{X}}_{\mathrm{T}_4}\end{bmatrix} 
= \begin{bmatrix}
0 & \mathrm{C}_{2} & \mathrm{C}_{3}\\
\end{bmatrix}\begin{bmatrix} \tilde{\mathrm{X}}_{\mathrm{T}_2}\\ \tilde{\mathrm{X}}_{\mathrm{T}_3}\\ \tilde{\mathrm{X}}_{\mathrm{T}_4}\end{bmatrix}\\
& \quad  = \begin{bmatrix}
0 & \tilde{\mathrm{V}}_{1} & \tilde{\mathrm{V}}_{2} \\
\end{bmatrix} \begin{bmatrix} \tilde{\mathrm{X}}_{\mathrm{T}_2}\\ \tilde{\mathrm{X}}_{\mathrm{T}_3}\\ \tilde{\mathrm{X}}_{\mathrm{T}_4}\end{bmatrix}= \begin{bmatrix}
\tilde{\mathrm{V}}_{1} & \tilde{\mathrm{V}}_{2}
\end{bmatrix} 
\begin{bmatrix}\tilde{\mathrm{X}}_{\mathrm{T}_3} \\ \tilde{\mathrm{X}}_{\mathrm{T}_4} \end{bmatrix}
 = \tilde{\mathrm{V}}\mathrm{X}_{\tilde{\mathrm{W}}},
\end{align*} noting that $c_2 = 4$ and $c_3 = 9$, and hence, $c_2+c_3 = 0$, $10c_2 = 1$, and $3c_3=1$.

\end{example}
\begin{example}
\normalfont 
Consider  a  scenario  where  the  server has $K=20$ messages, ${X}_1,\dots,{X}_{20}\in\mathbb{F}_{13}$, and the user wishes to compute $L=3$ linear combinations of $D=6$ messages ${X}_2,{X}_4,{X}_5,{X}_7,{X}_8,{X}_{10}$, say, 
\begin{align*}
Z_{1} & =7{X}_2+3{X}_4+12{X}_5+10{X}_7+2{X}_8+{X}_{10}	\\
Z_{2} & =3{X}_2+6{X}_4+5{X}_5+12{X}_7+8{X}_8+3{X}_{10}\\
Z_{3} & =5{X}_2+12{X}_4+{X}_5+4{X}_7+6{X}_8+9{X}_{10}
\end{align*}

For this example, the demand's support index set $\mathrm{W}$ is given by $\mathrm{W}=\{2,4,5,7,8,10\}$, and the demand's coefficient matrix $\mathrm{V}$ is given by 
\begin{equation*}
\mathrm{V} = 
\begin{bmatrix}
7 & 3 & 12 & 10 & 2 & 1\\
3 & 6 & 5 & 12 & 8 & 3\\
5 & 12 & 1 & 4 & 6 & 9
\end{bmatrix}.
\end{equation*} It is easy to verify that the matrix $\mathrm{V}$ is MDS, i.e., every $3\times 3$ submatrix of $\mathrm{V}$ is invertible (over $\mathbb{F}_{13}$). 

Let $\tilde{\mathrm{W}}$ be a sequence of length $6$ obtained by randomly permuting the elements in $\mathrm{W}$, for example, $\tilde{\mathrm{W}}=\{10,4,8,2,7,5\}$, 
and let $\tilde{\mathrm{V}}$ be a $3\times 6$ matrix obtained by applying the same permutation on the columns of the matrix $\mathrm{V}$, i.e., 
\begin{equation*}
\tilde{\mathrm{V}} = 
\begin{bmatrix}
1 & 3 & 2 & 7 & 10 & 12\\
3 & 6 & 8 & 3 & 12 & 5\\
9 & 12 & 6 & 5 & 4 & 1\\
\end{bmatrix},
\end{equation*} 
Note that $\mathrm{V}\mathrm{X}_{\mathrm{W}} = \tilde{\mathrm{V}}\mathrm{X}_{\tilde{\mathrm{W}}}$ by the construction of $\tilde{\mathrm{W}}$ and $\tilde{\mathrm{V}}$.



For this example, $R= K \pmod D =2$, $S=\gcd(D+R,R)=2$, $n=\lfloor \frac{K}{D}\rfloor-1=2$, and $m=\frac{R}{L}+1=\frac{5}{3}$. Note that for this example, $L=3>S=2$.


The query of the user consists of a $11\times20$ matrix $\mathrm{G}$ and a permutation $\pi$, constructed as follows. 
The matrix $\mathrm{G}$ consists of three submatrices (blocks) $\mathrm{G}_1$, $\mathrm{G}_2$, and $\mathrm{G}_3$ of size $3\times6$, $3\times 6$, and $5\times 8$, respectively, i.e., 
\begin{equation}\label{eq:17}
\mathrm{G} = 
\begin{bmatrix}
\mathrm{G}_1 & 0_{3\times 6} & 0_{3\times 8} \\
0_{3\times 6} & \mathrm{G}_2 & 0_{3\times 8} \\
0_{3\times 6} & 0_{3\times 6} & \mathrm{G}_3
\end{bmatrix},
\end{equation} where the construction of the blocks $\mathrm{G}_1$, $\mathrm{G}_2$, and $\mathrm{G}_3$ is described below.

The user randomly selects one of the blocks $\mathrm{G}_1, \mathrm{G_2},\mathrm{G}_3$, where the probability of selecting the block $\mathrm{G}_1$ (or the block $\mathrm{G}_2$) is $\frac{6}{20}$, and the probability of selecting the block $\mathrm{G}_3$ is $\frac{8}{20}$. 
Depending on $\mathrm{G}_1$, $\mathrm{G}_2$, or $\mathrm{G}_3$ being selected, the construction of each of these blocks is different. In this example, we consider the case that the user selects $\mathrm{G}_3$. In this case, the user takes $\mathrm{G}_1$ and $\mathrm{G}_2$ to be two randomly generated MDS matrices, each of size $3\times 6$, say, 
\begin{equation}\label{eq:18}
\mathrm{G}_{1} = \begin{bmatrix}
    11 & 5 & 3 & 1 & 4 & 2 \\
    7 & 10 & 2 & 6 & 6 & 5\\
    8 & 7 & 10 & 10 & 9 & 6 \\
\end{bmatrix}, \quad \quad 
\mathrm{G}_{2} = \begin{bmatrix}
    5 & 8 & 4 & 7 & 4 & 3\\
    7 & 4 & 12 & 9 & 1 & 10\\
    2 & 2 & 10 & 6 & 10 & 3\\
\end{bmatrix}.
\end{equation}

The user constructs $\mathrm{G}_{3}$ using a similar technique as in the Specialized MDS Code protocol of~\cite{EHS2021Joint}. The details of the construction of $\mathrm{G}_3$ are as follows. Recall that $\tilde{\mathrm{V}}$ generates a $[6,3]$ MDS code. Thus, the user can obtain the parity-check matrix $\Lambda$ of the code generated by $\tilde{\mathrm{V}}$ as 
\begin{equation*}
\mathrm{\Lambda} = 
\begin{bmatrix}
12 & 11 & 3 & 2 & 5 & 11\\
10 & 9 & 12 & 12 & 6 & 10\\
4 & 5 & 9 & 7 & 2 & 2
\end{bmatrix}.
\end{equation*}
Note that $\Lambda$ also generates a $[6,3]$ MDS code. Then, the user randomly chooses a $D=6$-subset of $\{1,\dots,8\}$, say, $\{h_1,\dots,h_6\} = \{1,3,4,6,7,8\}$, and randomly generates a $3\times 8$ MDS matrix $\mathrm{H}$ such that the submatrix of $\mathrm{H}$ restricted to columns indexed by $\{h_1,\dots,h_6\}= \{1,3,4,6,7,8\}$ (and all rows) is $\Lambda$. For this example, suppose that the user constructs the matrix $\mathrm{H}$ as
\begin{equation*}
\mathrm{H} = 
\begin{bmatrix}
\mathbf{12} & 4 & \mathbf{11} & \mathbf{3} & 3 & \mathbf{2} & \mathbf{5} & \mathbf{11} \\
\mathbf{10} & 7 & \mathbf{9} & \mathbf{12} & 4 & \mathbf{12} & \mathbf{6} & \mathbf{10} \\
\mathbf{4} & 9 & \mathbf{5} & \mathbf{9} & 1 & \mathbf{7} & \mathbf{2} & \mathbf{2}
\end{bmatrix}.
\end{equation*} Since $\mathrm{H}$ generates an $[8,3]$ MDS code, it can also be thought of as the parity-check matrix of a $[8,5]$ MDS code. The user then takes $\mathrm{G}_3$ to be the generator matrix of the $[8,5]$ MDS code defined by the parity-check matrix $H$, 
\begin{equation}\label{eq:19}
\mathrm{G}_{3} = \begin{bmatrix}
    \mathbf{1} & 4 & \mathbf{5} & \mathbf{9} & 2 & \mathbf{8} & \mathbf{4} & \mathbf{11}\\
    \mathbf{3} & 7 & \mathbf{10} & \mathbf{10} & 7 & \mathbf{9 }& \mathbf{10} & \mathbf{10}\\
    \mathbf{9} & 9 & \mathbf{7} & \mathbf{1} & 5 & \mathbf{2} & \mathbf{12} & \mathbf{2}\\
    \mathbf{1} & 6 & \mathbf{1} & \mathbf{4} & 11 & \mathbf{12} & \mathbf{4} & \mathbf{3}\\
    \mathbf{3} & 4 & \mathbf{2} & \mathbf{3} & 6 & \mathbf{7} & \mathbf{10} & \mathbf{11}
\end{bmatrix}.
\end{equation} 
Combining $\mathrm{G}_1$, $\mathrm{G}_2$, and $\mathrm{G}_3$ given by~\eqref{eq:18} and~\eqref{eq:19}, the user then constructs $\mathrm{G} $ as in~\eqref{eq:17}.

Next, the user constructs a permutation $\pi$ on $\{1,\dots,20\}$, using the procedure described in the step 1 of the protocol, and sends the permutation $\pi$ together with the matrix $\mathrm{G}$ to the server. 
Note that the columns $13,15,16,18,19,20$ of $\mathrm{G}$ are constructed based on the columns of $\Lambda$, the columns of $\Lambda$ are constructed based on the columns of $\tilde{\mathrm{V}}$, and the columns of $\tilde{\mathrm{V}}$ correspond to the messages $X_{10},X_{4},X_8,X_2,X_7,X_{5}$, respectively. 
Thus, the user constructs the permutation $\pi$ such that 
$\{\pi(10),\pi(4),\pi(8),\pi(2),\pi(7),\pi(5)\}=\{13,15,16,18,19,20\}$. 
For any $i\in \{1,\dots,20\}\setminus \{2,4,5,7,8,10\}$, the user then randomly chooses $\pi(i)$ (subject to the constraint that $\pi$ forms a valid permutation on $\{1,\dots,20\}$). For this example, suppose the user takes $\{\pi(1),\pi(3),\pi(6),\pi(9),\pi(11),\dots,\pi(20)\} = \{7,14,2,17,6,9,3,1,4,12,8,5,10,11\}$. 

Upon receiving the user's query, i.e., the matrix $\mathrm{G}$ and the permutation $\pi$, the server first permutes the components of the vector $\mathrm{X} = [X_1,\dots,X_{20}]^{\mathsf{T}}$ according to the permutation $\pi$ to obtain the vector $\tilde{\mathrm{X}}=\pi(\mathrm{X})$, i.e., $\tilde{X}_{\pi(i)} = X_i$ for $i\in \{1,\dots,20\}$. For this example, the vector $\tilde{\mathrm{X}}$ is given by
\[\tilde{\mathrm{X}} = [X_{14},X_6,X_{13},X_{15},X_{18},X_{11},X_{1},X_{17},X_{12},X_{19},X_{20},X_{16},X_{10},X_{3},X_{4},X_{8},X_{9},X_{2},X_{7},X_{5}]^{\mathsf{T}}.\]
Then the server computes $\mathrm{y=\mathrm{G}\mathrm{\tilde{\mathrm{X}}}}$, and sends the vector $\mathrm{y}$ back to the user as the answer. 

Let $\mathrm{T}_1 = \{1,\dots,6\}$, $\mathrm{T}_2 = \{7,\dots,12\}$, and $\mathrm{T}_3 = \{13,\dots,20\}$. 
For any $T\subset \{1,\dots,20\}$, we denote by $\tilde{\mathrm{X}}_{\mathrm{T}}$ the vector $\tilde{\mathrm{X}}$ restricted to its components indexed by $\mathrm{T}$. 
Note that $\mathrm{y} = [\mathrm{y}_1^{\mathsf{T}},\mathrm{y}_2^{\mathsf{T}},\mathrm{y}_3^{\mathsf{T}}]^{\mathsf{T}}$, where $\mathrm{y}_1 \triangleq \mathrm{G}_1 \tilde{\mathrm{X}}_{\mathrm{T}_1}$, $\mathrm{y}_2 \triangleq \mathrm{G}_2 \tilde{\mathrm{X}}_{\mathrm{T}_2}$, and $\mathrm{y}_3 \triangleq \mathrm{G}_3 \tilde{\mathrm{X}}_{\mathrm{T}_3}$. 
Then, the user recovers $[Z_1,Z_2,Z_3]^{\mathsf{T}} = \mathrm{V}\mathrm{X}_{\mathrm{W}} = \tilde{\mathrm{V}}\mathrm{X}_{\tilde{\mathrm{W}}}$ by computing
\begin{align*}
& \begin{bmatrix}
11 & 11 & 1 & 0 & 0\\
0 & 11 & 11 & 1 & 0\\
0 & 0 & 11 & 11 & 1	
\end{bmatrix}\mathrm{y}_3 
= 
\begin{bmatrix}
11 & 11 & 1 & 0 & 0\\
0 & 11 & 11 & 1 & 0\\
0 & 0 & 11 & 11 & 1	
\end{bmatrix} \mathrm{G}_3 \tilde{\mathrm{X}}_{\mathrm{T}_3}\\
& \quad =  
\begin{bmatrix}
11 & 11 & 1 & 0 & 0\\
0 & 11 & 11 & 1 & 0\\
0 & 0 & 11 & 11 & 1	
\end{bmatrix} 
\begin{bmatrix}
    \mathbf{1} & 4 & \mathbf{5} & \mathbf{9} & 2 & \mathbf{8} & \mathbf{4} & \mathbf{11}\\
    \mathbf{3} & 7 & \mathbf{10} & \mathbf{10} & 7 & \mathbf{9 }& \mathbf{10} & \mathbf{10}\\
    \mathbf{9} & 9 & \mathbf{7} & \mathbf{1} & 5 & \mathbf{2} & \mathbf{12} & \mathbf{2}\\
    \mathbf{1} & 6 & \mathbf{1} & \mathbf{4} & 11 & \mathbf{12} & \mathbf{4} & \mathbf{3}\\
    \mathbf{3} & 4 & \mathbf{2} & \mathbf{3} & 6 & \mathbf{7} & \mathbf{10} & \mathbf{11}
\end{bmatrix} \tilde{\mathrm{X}}_{\mathrm{T}_3}\\
& \quad = 
\begin{bmatrix}
\mathbf{1} & 0 & \mathbf{3} & \mathbf{2} & 0 & \mathbf{7} & \mathbf{10} & \mathbf{12}\\
\mathbf{3} & 0 & \mathbf{6} & \mathbf{8} & 0 & \mathbf{3} & \mathbf{12} & \mathbf{5}\\
\mathbf{9} & 0 & \mathbf{12} & \mathbf{6} & 0 & \mathbf{5} & \mathbf{4} & \mathbf{1}\\
\end{bmatrix} \begin{bmatrix} X_{10} \\ X_{3}\\ X_{4}\\ X_{8}\\ X_{9}\\ X_{2}\\ X_{7} \\ X_{5}\end{bmatrix}
= 
\begin{bmatrix}
1 & 3 & 2 & 7 & 10 & 12\\
3 & 6 & 8 & 3 & 12 & 5\\
9 & 12 & 6 & 5 & 4 & 1
\end{bmatrix}  \begin{bmatrix}X_{10} \\ X_{4}\\ X_{8}\\ X_{2}\\ X_{7} \\ X_{5}\end{bmatrix}  = \tilde{\mathrm{V}} \mathrm{X}_{\tilde{\mathrm{W}}}.
\end{align*}

\end{example}

\newpage

\bibliographystyle{IEEEtran}
\bibliography{PIR_PC_Refs}

\end{document}